\documentclass[11pt,letter]{article}
\usepackage{fullpage}
\usepackage[utf8x]{inputenc}
\usepackage{amsthm}
\usepackage{todonotes,wrapfig,float,graphicx,amssymb,textcomp,array,amsmath}
\usepackage{enumerate,enumitem}
\usepackage{multirow}
\usepackage{tabularx}
\usepackage{color,xcolor}
\usepackage{enumitem}
\usepackage{float}
\usepackage{nicefrac}

\allowdisplaybreaks

\definecolor{mycolor}{rgb}{0, 0, 0}

\usepackage{lineno}
\title{Polychromatic $2$-colorings with Bounded Discrepancy for Triangulations\thanks{Appeared in the 20th Scandinavian Symposium on Algorithm Theory (SWAT 2026).}}

\author{
Alma Arevalo Loyola\thanks{School of Computer Science, Carleton University, Canada \texttt{ALMAAREVALOLOYOLA@cmail.carleton.ca}}
\and Ahmad Biniaz\thanks{School of Computer Science, University of Windsor, Canada, \texttt{abiniaz@uwindsor.ca}. Research supported in part by NSERC.}
\and  Prosenjit Bose\thanks{School of Computer Science, Carleton University, Ottawa, Canada, \texttt{jit@scs.carleton.ca}. Research supported in part by NSERC.}
\and Thomas Shermer\thanks{School of Computing Science, Simon Fraser University, Canada, \texttt{shermer@sfu.ca}. Research supported in part by NSERC.}
}

\date{}
\newtheorem{lemma}{Lemma}
\newtheorem{corollary}{Corollary}
\newtheorem{proposition}{Proposition}

\newtheorem{theorem}{Theorem}
\newtheorem{observation}{Observation}
\newtheorem*{problem*}{Problem}
\newtheorem*{claim*}{Claim}
\newtheorem*{invariant*}{Invariant}

\newtheorem{definition}{Definition}

\usepackage{graphicx}
\usepackage{mathtools}
\usepackage{amsthm}
\usepackage{algorithm}
\usepackage{algorithmic}
\usepackage{multicol}
\usepackage{enumerate} 
\usepackage{subcaption}
\newcommand{\remove}[1]{}

\begin{document}
\maketitle
\begin{abstract}
A polychromatic $2$-coloring of a triangulation is a $2$-coloring of the vertices such that no face is monochromatic. The discrepancy of a
	coloring is the maximum difference between the sizes of the color
	classes. Asayama and Matsumoto (Graphs and Combinatorics, 2022) proved
	that every triangulation admits a polychromatic $2$-coloring with
	discrepancy at most $\tfrac{5n-16}{9}$, and that there exists a class
	of triangulations for which every polychromatic $2$-coloring has
	discrepancy at least $\tfrac{n}{3} - 2$, where $n$ is the number of
	vertices. We improve the upper bound, showing that every triangulation
	admits a polychromatic $2$-coloring with discrepancy at most
	$\tfrac{3n-16}{7}$ and such a $2$-coloring can be computed in quadratic
	time. We also show a discrepancy of at most $n-\tfrac{4M}{3}$ for
	triangulations with a matching of size $M$. This implies, for example,
	that Delaunay triangulations admit a discrepancy of at most
	$\tfrac{n}{3}$. We provide a linear-time algorithm to compute a
	$2$-coloring whose discrepancy is at most $\tfrac{5n-24}{7}$.

One of our results shows that any proper four coloring with the largest color class of size $\frac{n}{2}$ would imply a $2$-coloring with discrepancy at most $\frac{n}{3}$. The existence of such a proper coloring has been recently confirmed by Kawarabayashi, Yoneda, and Yoneda (arXiv 2026). Therefore the two results together confirm the discrepancy of at most $\frac{n}{3}$ for triangulations.
\end{abstract}
\section{Introduction}
Colorings of planar graphs play a central role in both combinatorial and algorithmic graph theory. The study of colorings that ensure coverage or visibility conditions has received growing attention in recent years, particularly in contexts where each region of a planar subdivision must contain representatives from multiple categories or sensors. For instance, such colorings arise naturally in network coverage problems, geometric guarding, and frequency assignment, where it is desirable that every bounded region has access to all available resources. This motivated the study of \emph{polychromatic colorings}, colorings in which every face of a planar embedding contains at least one vertex of each color. For a survey on polychromatic colorings see \cite{czap2017facially}.

For the case of two colors, a polychromatic $2$-coloring of a triangulation assigns one of two colors to each vertex such that no face is monochromatic. Unlike proper colorings, which aim to separate adjacent vertices, polychromatic colorings of planar graphs emphasize inclusion: every face must {\em see} both colors. Although such colorings always exist, the balance
between the sizes of the two color classes is not guaranteed. 
The {\em discrepancy} of a specific polychromatic $2$-coloring of $G$ is defined as the absolute difference between the sizes of its two color classes. It measures how unevenly the two colors are distributed. For a graph $G$, we define $\mathrm{disc}(G)$ as the minimum discrepancy over all possible polychromatic $2$-colorings of $G$. Thus, $\mathrm{disc}(G)$ captures the best achievable balance between the two colors while maintaining the polychromatic property. A polychromatic coloring of $G$ is said to be \emph{balanced} when its discrepancy is at most one.

Asayama and Matsumoto~\cite{asayama2022balanced} initiated the systematic study of balanced polychromatic $2$-colorings in planar triangulations. Among other results, they proved that every triangulation on $n$ vertices admits a polychromatic $2$-coloring with discrepancy at most $\tfrac{5n-16}{9}$, and they constructed infinite families of triangulations where every such coloring has discrepancy at least $\tfrac{n}{3}-2$. They conjectured that $\tfrac{n}{3}$ is the best possible asymptotic upper bound.

In this work we make progress toward that conjecture. We show that every triangulation on $n$ vertices admits a polychromatic $2$-coloring satisfying
$\mathrm{disc}(G)\le \tfrac{3n-16}{7}$,
which improves the previous bound of $\tfrac{5n-16}{9}$. Our proof refines the discharging arguments used in earlier analyses and integrates structural constraints derived from the Four Color Theorem. Moreover, the proof is constructive and leads to an $O(n\log n)$ algorithm to compute such a coloring.

We further establish that if a triangulation $G$ admits a matching of size $M$, then $G$ has a polychromatic $2$-coloring with discrepancy at most $\tfrac{n-4M}{3}$. This bound is tight for certain triangulations and immediately implies that Delaunay triangulations satisfy $\mathrm{disc}(G)\le \frac{n}{3}$, since they always satisfy having a perfect matching. In addition, we provide a linear-time recoloring procedure that guarantees a polychromatic $2$-coloring with discrepancy at most $\tfrac{5n-24}{7}$. This yields the first efficient algorithm that achieves a nontrivial upper bound on discrepancy in planar triangulations.

Towards our proof of the $\frac{3n-16}{7}$ we prove (in Lemma~\ref{lemma:n_2}) that if a triangulations $G$ admits a proper four coloring with the largest color class of size $\frac{n}{2}$ then $\mathrm{disc}(G)\le \frac{n}{3}$. After the first public appearance of our paper in SWAT 2026, the existence of such a proper coloring was confirmed by Kawarabayashi, Yoneda, and Yoneda \cite{Kawarabayashi2026}.  Therefore the two results together confirm the discrepancy of at most $\frac{n}{3}$. Although this new bound supersedes our previous $\frac{3n-16}{7}$ bound, we decided to keep its proof because it contains elements that are interesting on their own. In particular our vertex classification to essential and non-essential vertices based on alternating separating cycles, as well as the color swapping technique (both presented in Section~\ref{improve-sec}) could be useful for similar problems.

Our results strengthen the connection between structural graph properties and color balance in planar settings, advancing the quantitative understanding of how evenly a triangulation can be $2$-colored while ensuring that every face remains non-monochromatic. 

\section{Preliminaries}

Most of the basic notions about graph coloring can be found in \cite{jensen2011graph} and \cite{tuza2013section}. For general concepts of graph theory, see \cite{bondy2008graph} and \cite{diestel2025graph}. We consider graphs that are finite and simple; that is, they have no loops or multiple edges. We denote a graph by $ G = (V, E)$, where $V$ is the set of vertices, and $E$ is the set of edges.  For planar graphs, we denote by $F$ the set of faces. When needed, we denote by $V(G)$, $E(G)$ and $F(G)$ the sets of vertices, edges and faces to specify the graph to which we refer. We use $n$, $m$, and  $l$ to denote the number of vertices, edges and faces, respectively. The \emph{degree} of a vertex $v \in V$ is the number of vertices adjacent to $v$, and it is denoted by $d(v)$. We denote by $\Delta(G)$ and $\delta(G)$ the maximum and minimum degree of the vertices of $G$, respectively. The degree of a face $f \in F$, denoted by $d(f)$, is the number of edges (or vertices) on its boundary.
A \emph{vertex coloring} $\chi$ of $G$ is a function $\chi: V \rightarrow \{c_1, c_2, ... , c_k\}$. We call it a \emph{$k$-coloring} to emphasize the number of colors. 

\begin{definition}
    Let $G$ be a graph with $n$ vertices, let $\chi$ be a vertex coloring of $G$, and let $\{V_1, V_2, \dots, V_k \}$ be the partition induced by $\chi$. We define the \emph{discrepancy} of $\chi$, denoted as $\operatorname{disc}(\chi)$, to be the difference between the sizes of the largest and the smallest color class:  $$\operatorname{disc}(\chi) = \max\{|V_i| - |V_j| : i, j \in \{ 1, 2, \dots, k\}\}.$$
\end{definition}
We say that a coloring $\chi$ is \emph{balanced} when $\operatorname{disc}(\chi) \leq 1$. Balanced proper colorings are also called \emph{equitable colorings} \cite{meyer1973equitable}.

The following are some known results that we use in the following sections. 
  
\begin{proposition}[Euler’s Formula]
Let $G$ be a connected planar graph with $n$ vertices, $m$ edges, and $l$ faces.
Then $n - m + l = 2$.  
\end{proposition}

%\alma{add reference}
\begin{proposition}[Diestel \cite{diestel2025graph}]\label{propEuler}
Every planar graph with $n \geq 3$ vertices has at most $3n-6$ edges and at most $2n-4$ faces. Every planar triangulation with $n \geq 3$ vertices has exactly $3n-6$ edges and $2n-4$ faces. Every planar quadrangulation with $n \geq 3$ vertices has exactly $2n-4$ edges and $n-2$ faces. 
\end{proposition}

% The concept of independence is relevant in the study of coloring and it is used in several proofs in this paper. 
An \emph{independent set} $I$ of $G$ is a subset of vertices such that no two elements of $I$ are adjacent. Observe that in a proper coloring, every color class is an independent set. The \emph{independence number} of $G$, denoted by $\alpha(G)$, is the size of the maximum independent set.

%Theorem 7. Independence Number of Maximal Planar Graphs, Allan Bickle
\begin{proposition}[Caro and Roditty \cite{caro1985vertex}]\label{propMIS}
Let $G$ be a planar triangulation with $n \geq 4$ vertices and
minimum degree $\delta(G)$. Then $\alpha(G) \leq \frac{2n -4}{\delta(G)}.$
\end{proposition}

Given a planar graph $G$ and its planar embedding, the \emph{dual graph} of $G$, denoted by $G^*$ is such that the set of vertices of $G^*$ corresponds to the set of faces of $G$, and two vertices are adjacent if and only if the corresponding two faces share an edge of their boundary. The following claim is implied by Tait's reformulation of the Four Color Theorem \cite{Tait_1880}.

 \begin{proposition}
     Let $G$ be a planar triangulation. The dual graph $G^*$ has a proper $3$-edge-coloring. 
 \end{proposition}

A \emph{matching} $M$ of $G$ is a subset of edges such that no two elements of $M$ share a common vertex. A \emph{maximum matching} is a matching with maximum cardinality. A \emph{perfect matching} is a matching that includes all vertices of $G$, and a \emph{near-perfect matching} includes all of them, but one. 

\begin{proposition}[Nishizeki and Baybars \cite{nishizeki1979lower}]\label{propMatching}
    Every planar triangulation $G$ with $n$ vertices has a matching of size at least $\frac{n}{3}$.
\end{proposition}

\section{Upper Bounds}

In this section, we prove that every planar triangulation has a polychromatic $2$-coloring with discrepancy at most $\frac{3n- 16}{7}$. We first show some bounds related to other parameters of $G$. 

In the first subsection, given a planar triangulation $G$ with a maximum matching $M$, we show how to find a polychromatic $2$-coloring whose discrepancy is at most $n-\frac{4|M|}{3}$.  
In particular, the discrepancy is at most $\frac{n}{3}$ when $G$ has a perfect matching. 

In the second subsection, given a proper $4$-coloring of the triangulation $G$ we show how to partition the four color classes into two sets so that the resulting coloring is polychromatic and its discrepancy is bounded depending on the size of the largest color class in the $4$-coloring. In particular, we guarantee the existence of a polychromatic $2$-coloring with discrepancy at most $\frac{n}{3}$ for graphs with independent number at most $\frac{n}{2}$.
In the last subsection, we extend this idea to improve the general upper bound.

\subsection{Discrepancy and Matchings}
Polychromatic $2$-colorings of triangulations are closely related to spanning quadrangulations and matchings as stated in \cite{asayama2022balanced}, \cite{biedl2001efficient}, and \cite{bose2003worst}. It is known that every triangulation has a spanning quadrangulation, where a {\em quadrangulation} is a plane graph such that each of its faces is a quadrilateral. It is also known that every quadrangulation is bipartite, and hence $2$-colorable.

Let $G$ be a planar triangulation, and let $Q$ be a spanning quadrangulation of $G$. Every face of $Q$ is the union of two adjacent faces of $G$, with their shared edge removed. Moreover, $Q$ contains exactly two boundary edges of each face of $G$. Observe that the unique proper $2$-coloring of $V(Q)$ corresponds to a polychromatic $2$-coloring of $V(G)$.  
We now show how the size of the discrepancy relates to the size of a matching of $Q$ in the following lemma.

\begin{lemma} \label{lem:disc}
 Let $\rho$ be the unique proper $2$-coloring of $Q$ and let $M_Q$ be a matching in $Q$. Then, $\rho$ is a polychromatic $2$-coloring in $G$ with $\operatorname{disc}(\rho) \leq n - 2|M_Q|$.
\end{lemma}

 \begin{proof}

Let $R$ and $B$ be the color classes in $\rho$.  Every edge in the matching $M_Q$ has one endpoint in each color class. Let $U_R \subseteq R$ and $U_B \subseteq B$ be the unmatched vertices. Then $|R| = |M_Q| + |U_R|$ and $|B| = |M_Q| + |U_B|$. Without loss of generality, assume that $|R| \geq |B|.$
Then, $\operatorname{disc}(\chi) = |R|-|B| = |U_R| - |U_B|\leq |U_R| + |U_B| = n-2|M_Q|$.
 \end{proof}

Now, to define an appropriate spanning quadrangulation with the appropriate matching, the strategy is to start with a matching $M$ in the triangulation $G$, and then show that there exists a spanning quadrangulation containing a large fraction of edges from $M$.

 There is a well-known equivalence between the $4$-coloring of the vertices of a triangulation $G$ and the $3$-coloring of the edges of its dual $G^*$. Such edge-coloring induces a partition of the edges of $G^*$ into 3 perfect matchings of $G^*$. Let $E_1$, $E_2$ and $E_3$ be the sets of edges of $G$ corresponding to each of the 3 perfect matchings in $G^*$. 
 Let $Q_i = G\setminus\{E_i\}$ for $i\in \{1,2,3\}$. We note that each $Q_i$ is a spanning quadrangulation of $G$.  (see e.g.~\cite{asayama2022balanced}, \cite{bose2003worst}, \cite{bose1997guarding}, and \cite{kundgen2017spanning}).

\begin{figure}[ht!]
	\centering
	\setlength{\tabcolsep}{0in}
	$\begin{tabular}{ccc}
		\multicolumn{1}{m{.33\columnwidth}}{\centering\vspace{0pt}\includegraphics[width=.32\columnwidth]{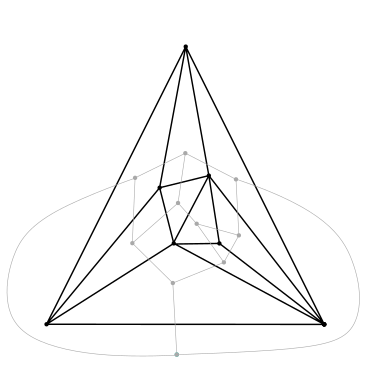}}
		&\multicolumn{1}{m{.33\columnwidth}}{\centering\vspace{0pt}\includegraphics[width=.32\columnwidth]{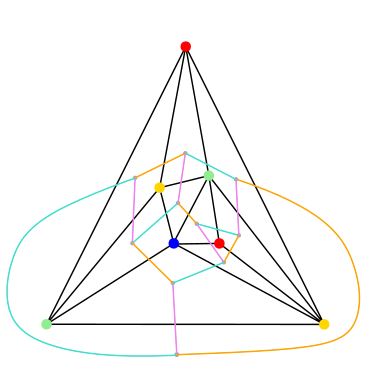}}
        &\multicolumn{1}{m{.33\columnwidth}}{\centering\vspace{0pt}\includegraphics[width=.32\columnwidth]{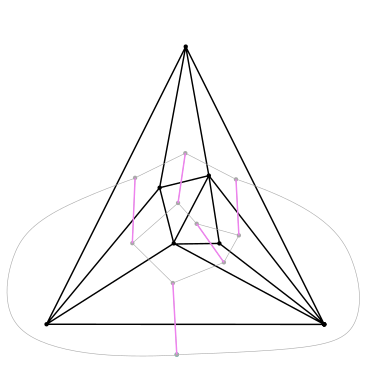}}
		\\
		(a)   &(b) &(c)\\
        \multicolumn{1}{m{.33\columnwidth}}{\centering\vspace{0pt}\includegraphics[width=.32\columnwidth]{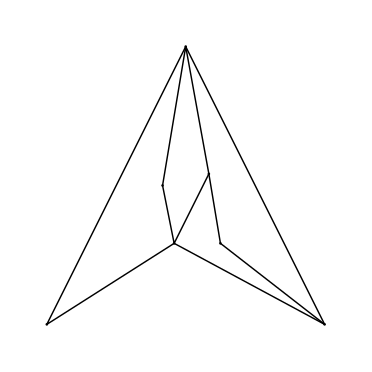}}
		&\multicolumn{1}{m{.33\columnwidth}}{\centering\vspace{0pt}\includegraphics[width=.32\columnwidth]{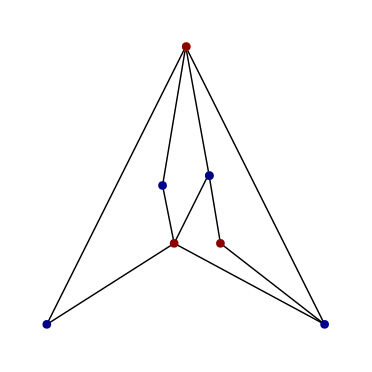}}
        &\multicolumn{1}{m{.33\columnwidth}}{\centering\vspace{0pt}\includegraphics[width=.32\columnwidth]{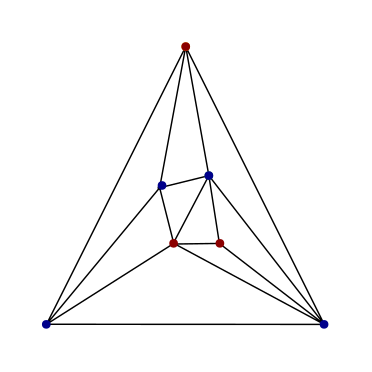}}
		\\
		(d)   &(e) &(f)
	\end{tabular}$	

     \caption{(a) A triangulation $G$ and its dual graph $G^*$. (b) Equivalence between a $4$-coloring on $V(G)$ and a $3$-coloring on $E(G)$. (c) Perfect matching on $G^*$, corresponding to one color class in the 3-edge-coloring. (d) Spanning quadrangulation of $G$ induced by deleting the dual edges of the perfect matching on $G^*$. (e) Proper 2-coloring on $V(Q)$. (f) Polychromatic 2-coloring on $V(G).$} 
   \label{fig:foobar}

\end{figure}

\begin{lemma}\label{lem:23} Let $M$ be a matching in $G$. There exists a spanning quadrangulation of $G$ that has a matching of size at least $\frac{2}{3}|M|$.
\end{lemma}

 \begin{proof}
     Every edge of $M$ appears by construction in exactly 2 of the $Q_i$'s, which implies that one of the $Q_i$'s contains at least $\frac{2}{3}|M|$ edges of $M$.
 \end{proof}

The following is the main result of this subsection.

\begin{theorem}\label{thm:disc_matching}
Every triangulation $G$ with $n$ vertices has a polychromatic $2$-coloring with discrepancy at most $n-\frac{4}{3}|M|$, where $M$ is
a matching in $G$.
\end{theorem}

\begin{proof}
 By Lemma~\ref{lem:23}, $G$ has a a spanning quadrangulation $Q$ that has a matching $M_Q$ of size at least $\frac{2}{3}|M|$. By Lemma~\ref{lem:disc} the proper $2$-coloring of $Q$ is a polychromatic $2$-coloring of $G$ with discrepancy at most   $n-2|M_Q|\leq n-\frac{4}{3}|M|$.
\end{proof}

\begin{corollary}\label{matchings}
 Every planar triangulation $G$ has a polychromatic $2$-coloring  whose discrepancy is at most $\frac{5n - 16}{9}$.   
\end{corollary}
\begin{proof}

Follows from Theorem \ref{thm:disc_matching} with the matching that is   guaranteed by Proposition \ref{propMatching}.
%By Proposition \ref{propMatching} and Theorem \ref{thm:disc_matching}, $G$ has a polychromatic $2$-coloring whose discrepancy is at most $n - \frac{4}{3}\left(\frac{n}{3}\right) = n - \frac{4n}{9} = \frac{5n}{9}$. 
\end{proof}

A relevant application of Theorem \ref{thm:disc_matching} is to obtain a better upper bound for the discrepancy of polychromatic $2$-colorings in Delaunay triangulations, stated in the following Corollary. 

\begin{corollary}
Let $D$ a graph on $n$ vertices corresponding to the Delaunay triangulation of a set of $n$ points in the plane. $D$ has a polychromatic $2$-coloring whose discrepancy is at most $\frac{n}{3}$.
\end{corollary}

\begin{proof}
    It is known that Delaunay triangulations have perfect matchings \cite{Biniaz2021,dillencourt1990}. Then by Theorem \ref{thm:disc_matching}, $D$ has a polychromatic $2$-coloring whose discrepancy is at most $n - \frac{4}{3}\left(\frac{n}{2}\right)$. \end{proof}

Observe that the bound $\frac{n}{3}$ is true for any triangulation with a perfect matching.

 %\alma{añadir cosas de Delaunay en Preliminaries}

\subsection{Discrepancy and Independent Sets}\label{sub:independent}

In this section, the strategy is to start with a proper $4$-coloring of the triangulation, and define a $2$-coloring that is polychromatic whose discrepancy is bounded according to the sizes of the color classes in the $4$-coloring.
This idea was used in \cite{asayama2022balanced} to prove the general bound. We extend the idea to relate the discrepancy to the size of the maximum independent set, which generalizes two of their main results and Theorem \ref{thm:disc_matching}.
%Esto va a ser usado para lo que sigue:

Let $G$ be a planar triangulation with $n$ vertices.  Let $\chi: V(G) \rightarrow \{1, 2, 3, 4\}$ be a $4$-proper coloring of $V(G)$. Let $V_1, V_2, V_3, V_4$ be the respective color classes and let $n_1$, $n_2$, $n_3$, $n_4$ be their respective sizes.
We note that $n_1 + n_2 + n_3 + n_4 = n$. Without loss of generality, assume $n_1 \geq n_2 \geq n_3 \geq n_4$. Since $n_2 + n_3 + n_4  =  n - n_1$ we get the following inequality:

\begin{observation}\label{obs:smallerclass} $n_4  \leq  \frac{n - n_1}{3}.$
\end{observation}

We describe a $2$-coloring $\rho$ of $V(G)$, whose color classes are $R$, $B$ (sometimes referred to as the colors red and blue), and $r$, $b$ their respective sizes. The discrepancy of $\rho$ is defined as  $\operatorname{disc}(\rho) = |r - b|$ with $r + b = n $. Therefore, $\operatorname{disc}(\rho) = |2r - n|$.  We consider two cases to eliminate the absolute value.

\begin{observation}\label{obs:disc_twocases}

If $r \geq \frac{n}{2}$, then 
    $\operatorname{disc}(\rho)  =    2r - n .$
If $r < \frac{n}{2}$, then $\operatorname{disc}(\rho)  =   n - 2r.$
    % \begin{itemize}
    %     \item If $r \geq \frac{n}{2}$, then 
    % $\operatorname{disc}(\rho)  =    2r - n .$
    % \item  If $r < \frac{n}{2}$, then $\operatorname{disc}(\rho)  =   n - 2r.$
    % \end{itemize}

\end{observation}

\begin{lemma}\label{lemma:n_2}
If $n_1 + n_4 \le \frac{2n}{3}$ or $n_1 \le \frac{n}{2}$, then $G$ has a polychromatic $2$-coloring whose discrepancy is at most $\frac{n}{3}$.
\end{lemma}
 \begin{proof}
 Let $\rho$ be the $2$-coloring that assigns $R = V_1 \cup V_4$ and $B = V_2 \cup V_3$. Then $r = n_1 + n_4$ and $ b= n_2 + n_3$. 
We claim that $\rho$ is a polychromatic $2$-coloring. Since every face of $G$ is triangular, it contains vertices from 3 different color classes of the $4$-coloring $\chi$.  
%Every color class of $\rho$ is the union of two color classes of the proper $4$-coloring. 
Therefore, in $\rho$ every face contains at least one red and one blue vertex. 

 Since $n_1 + n_4 \geq n_2 \geq n_3$, we have that 
 $r = n_1 + n_4 \geq \frac{n}{3}$. 
  Thus, if $n_1 + n_4 \le \frac{2n}{3}$ then $\operatorname{disc}(\rho)  = |2r-n|\le \frac{n}{3}$.

 If $n_1\le\frac{n}{2}$ then by applying Observation~\ref{obs:smallerclass} we get $n_2+n_3=n-n_1-n_4\ge \frac{2n-2n_1}{3}\ge \frac{n}{3}$. Once again we get then $\operatorname{disc}(\rho)  = |2r-n|\le \frac{n}{3}$.
 
% & \leq & \frac{n}{3}$
%  \begin{eqnarray*}
% \operatorname{disc}(\rho') & \leq & n- 2\left(\frac{n}{3}\right)\\
% & \leq & \frac{n}{3}.
% \end{eqnarray*}
 \end{proof}

\begin{lemma}\label{lemma:n_1}
If $n_1 + n_4 \geq \frac{n}{2}$, then $G$ has a polychromatic $2$-coloring whose discrepancy is at most $\frac{4n_1 - n}{3}$.
\end{lemma}

\begin{proof}

 Let $\rho$ be the polychromatic 2-coloring defined in the proof of Lemma \ref{lemma:n_2}.
By observation \ref{obs:smallerclass}, we know that $r  = n_1 + n_4  \leq  \frac{n + 2n_1}{3}$.
% \begin{equation}\label{eqn:upper}
% r  = n_1 + n_4  \leq  \frac{n + 2n_1}{3}\\
% \end{equation} 
Since $r = n_1 + n_4 \geq \frac{n}{2}$, we have $\operatorname{disc}(\rho)  =    2r - n .$ We conclude that \begin{equation}\notag\operatorname{disc}(\rho)  \leq   2\left(\frac{n + 2n_1}{3}\right) - n \leq  \frac{4n_1 - n}{3}.\end{equation}
\end{proof}

% \begin{eqnarray*}
% \operatorname{disc}(\rho') & \leq &  2\left(\frac{n + 2n_1}{3}\right) - n\\
% & \leq & \frac{4n_1 - n}{3}.
% \end{eqnarray*}

\begin{theorem}\label{thmDiscrepancy}
Let $G$ be a planar triangulation with $n$ vertices. If $G$ has a proper $4$-coloring $\chi$ whose largest color class has size at most $s$, then $G$ has a polychromatic $2$-coloring $\rho$ such that $\operatorname{disc}(\rho) \leq \max\left(\frac{n}{3}, \frac{4s - n}{3}\right)$.   
\end{theorem}

\begin{proof}
The sizes of the color classes in $\chi$ meet the conditions of Lemma \ref{lemma:n_2} or \ref{lemma:n_1}.
\end{proof}

% Observe that Theorem \ref{thmDiscrepancy} is a generalization of two previous results (\emph{Theorem 7 and Theorem 4 in \cite{asayama2022balanced}}), which can be condensed in the following corollary.
\begin{observation}\label{obs:independentnumber}
  By the Four Color Theorem, every triangulation has a proper $4$-coloring. Since the color classes are independent sets, Theorem \ref{thmDiscrepancy} can always be applied with $s = \alpha(G)$, the independence number of $G$.  
\end{observation}

\begin{corollary}\label{corollaryUpper}
      Let $G$ be a planar triangulation with $n \geq 4$ vertices. Then $G$ has a polychromatic $2$-coloring whose discrepancy is at most $\frac{5n-16}{9}$. Moreover, if $\delta(G) \geq 4$ then $G$ has a polychromatic $2$-coloring whose discrepancy is at most $\frac{n}{3}.$

\end{corollary}

\begin{proof}
By the Four Color Theorem, $G$ has a proper $4$-coloring $\chi$. Let $n_1$ be the size of the largest color class of $\chi$. 
%Since the color classes of $\chi$ are independent sets, we have $n_1 \leq \alpha(G) \leq \frac{2n-4}{\delta(G)}$, 
By Observation \ref{obs:independentnumber} and Proposition \ref{propMIS}, we have $n_1 \leq \alpha(G) \leq \frac{2n-4}{\delta(G)}$ and we can apply Theorem \ref{thmDiscrepancy} with $s =\frac{2n-4}{\delta(G)}$. Since every planar triangulation with $n \geq 4$ vertices has minimum degree at least $3$, it follows that:
\[\operatorname{disc}(\rho) \leq \max\left(\frac{n}{3}, \frac{4\left(\frac{2n-4}{3}\right) - n}{3}\right) = \frac{5n - 16}{9}.\]

If $\delta(G) \geq 4$ then \[\operatorname{disc}(\rho) \leq \max\left(\frac{n}{3}, \frac{4\left(\frac{2n-4}{\delta(G)}\right) - n}{3}\right) = \frac{n}{3}.\qedhere\]

\end{proof}

There is a relation between the size of a matching and the size of the maximum independent set in $G$. 
\begin{observation}
   Given a matching $M$ of $G$, any subset of more than $n - |M|$ vertices includes an edge of $M$ and then it is not independent. Therefore, $\alpha(G) \leq n - |M| $ for any matching $M$. Based on this fact and Observation \ref{obs:independentnumber}, we can apply Theorem \ref{thmDiscrepancy} with $s = n - |M|$ and get Theorem \ref{thm:disc_matching}.
\end{observation}

We have provided an alternative proof for two theorems from \cite{asayama2022balanced}.
%Theorem 7 and Theorem 4 
Moreover, we have generalized these results by establishing an upper bound on the discrepancy of a polychromatic $2$-coloring in terms of an upper bound on the size of the largest color class in a proper $4$-coloring.

\subsection{Improving the General Discrepancy Bounds}
\label{improve-sec}

In this section we prove the following theorem:

\begin{theorem}\label{thm:3/7}
 Let $G$ be a planar triangulation with $n \geq 6$ vertices. Then $G$ has a polychromatic $2$-coloring whose discrepancy is at most $\frac{3n-16}{7}$. 
\end{theorem}

Recall the two color classes from the previous section (that are obtained as a partition of the colors of a proper $4$-coloring). In view of Lemma~\ref{lemma:n_2} we may assume that $n_1+n_4\ge n_2+n_3$ as otherwise the discrepancy is at most $n/3$. 
The main proof idea is to start by coloring all vertices in $V_1$ red, all vertices in $V_2 \cup V_3$ blue, and possibly splitting the vertices in $V_4$ between red and blue to get a smaller discrepancy. 

Let $F_1 \subset F(G)$ be the set of faces that contain a vertex from $V_1$ and let $F_{234} = F(G) \setminus F_1 $ be the set of faces that do not.
%contain a vertex from $V_1$. 
If $V_1 \subseteq R $ and $ V_2 \cup V_3 \subseteq B$ then all faces in $F_1$ are guaranteed to be polychromatic, regardless of the color of the vertices in $V_4$. In addition, every face $f \in F_{234}$ contains %two vertices in $B$ and 
a vertex from $V_4$ that we can assign to $R$ to make $\rho$ polychromatic. Let $S_4 \subseteq V_4$ be the set of vertices in $V_4$ belonging to at least one face in $F_{234}$. Let $s_4 = |S_4|$. We call the vertices in $S_4$ \emph{essential}, since they are required to be in the color class $R$ to make the coloring polychromatic. 
The rest of the vertices in $V_4$ are called \emph{non-essential}, and we denote the set as $\bar{S}_4 = V_4 \setminus S_4$.  
Let $\rho$ be the $2$-coloring that assigns $R= V_1 \cup S_4$ and let $B =V_2 \cup V_3 \cup \bar{S}_4$.  The sizes of the color classes of $\rho$ are $r = n_1 + s_4$ and $b = n_2 + n_3 + n_4 - s_4$. Observe that $\rho$ is a polychromatic $2$-coloring and $disc(\rho) = 2r-n$, by Observation \ref{obs:disc_twocases} and Lemma~\ref{lemma:n_2}. 
%Because $r \geq \frac{n_1} \geq \frac{n}{2}$.
As a warm-up, in the following proposition we bound the discrepancy of $\rho$ by analyzing the properties of $S_4$; this result is independent of the rest of the section.

%\alma{name $F_{234}$? define good faces?}
  % Since every face of $G$ is triangular, it contains three different colors from the proper $4$-coloring $\chi$. Notice that every face $f \in F(G)$ has at least one vertex in $R = V_1 \cup S_4$, because $f$ is either in $F_1$ or it has a vertex in $S_4$. On the other hand, every face contains at least one vertex in $B = V_2 \cup V_3$.
  
% \begin{theorem}
% Every planar triangulation $G$ with n vertices has a polychromatic $2$-coloring $\rho$ such that $\operatorname{disc}(\rho) \leq \frac{n}{2}$.
% \end{theorem}

  \begin{proposition}
  $disc(\rho) \leq \frac{n-4}{2}$.
  \end{proposition}

  \begin{proof}
  
We give an upper bound for $r = n_1 + s_4$ based on the following observations:
\begin{enumerate}
    
\item 
Every vertex in $V_1$ is in at least $3$ faces from $F_1$ because $\delta(G)\geq 3$. Every face in $F_1$ has exactly one vertex from $V_1$. Then $ 3n_1 \leq |F_1| $.

\item 
Every vertex in $S_4$ is in at least one face from $F_{234}$ by definition. Every face in $F_{234}$ has exactly one vertex in $S_4$. Then $s_4 \leq |F_{234}| $.
 
\end{enumerate}
 By Proposition \ref{propEuler}, the total number of faces of $G$ is $|F_1|+|F_{234}| = 2n - 4$, then we conclude that:

\begin{equation}\label{eqn:faces}
3n_1 +s_4\leq |F_1| + |F_{234}| =  2n-4 
\end{equation}

We know that $S_4 \subseteq V_4$,  then $s_4 \leq n_4$. By Observation \ref{obs:smallerclass}, we have $ n_4  \leq \frac{n - n_1}{3}$, which implies that:
\begin{equation}\label{eqn:classes}
n_1 + 3s_4 \leq n_1 + 3\left(\frac{n - n_1}{3}\right) = n
\end{equation}

By adding up the equations \ref{eqn:faces} and \ref{eqn:classes}, we get $4n_1 + 4s_4  \leq 3n - 4 $. Then $r = n_1 + s_4 \leq \frac{3n-4}{4}$
%By Observation \ref{obs:disc_twocases}, 
and $\operatorname{disc}(\rho) = 2r - n \leq \frac{n-4}{2}$.
\end{proof}

Some extra observations help to improve this bound. 
Let $G$ be a planar triangulation with $n$ vertices.  Let $\chi_{min}$ be a proper $4$-coloring of $V(G)$ such that the difference between the sizes of the largest and the smallest color class, i.e. $n_1-n_4$, is minimum over all proper $4$-colorings of $V(G)$. From the results in Subsection \ref{sub:independent} (Lemma~\ref{lemma:n_2}), we know that if $n_1 + n_4 \leq \frac{2n}{3}$ or $n_1 \leq \frac{n}{2}$, then the coloring that assigns $R= V_1 \cup V_4$ and $B = V_2 \cup V_3$ is polychromatic and has a discrepancy at most $\frac{n}{3}$. 
 % Also, if $n_1 \leq \frac{n}{2}$ then Theorem \ref{thmDiscrepancy} guarantees that a polychromatic $2$-coloring $\rho$ exists such that $\operatorname{disc}(\rho) \leq \frac{n}{3}$. \alma{(add the properties we can assume by previous result)} From the previous subsection, we know that if $n_1 \leq \frac{n}{2}$ then Theorem \ref{thmDiscrepancy} guarantees that a polychromatic $2$-coloring $\rho$ exists such that $\operatorname{disc}(\rho) \leq \frac{n}{3}$.
Thus, we may assume that $n_1 > \frac{n}{2}$ and $n_1 +n_4 > \frac{2n}{3}$ in the following. Then, Observation \ref{obs:smallerclass} implies that $n_4\le\frac{n}{6}$. Therefore, we get 

 \begin{observation}\label{n1n4}
     $n_1 - n_4 \geq \frac{n}{3}$.
     %$diff(\chi_{min}) =
 \end{observation}

Let $V_1^{d=3}$ be the set of vertices in $V_1$ whose degree equals $3$ and let $V_1^{d\geq4} = V_1 \setminus V_1^{d=3}$ be those with degree at least $4$.

% Consider the partition of $V_1$ into vertices of degree $3$ and vertices of degree at least $4$. This induces a partition of the set $F_1$ because every face of $F_1$ contains exactly one vertex from $V_1$. 
% Let $\bar{F_1}$ be the set of faces that contain a vertex from $V_1$ whose degree is exactly $3$. Then $F_1 \setminus \bar{F_1}$ corresponds to the set of faces that contain a vertex in $V_1$ whose degree is at least $4$.
% %explicar?
%Let $F_{234} = F(G) \setminus F_1$, the set of faces without a vertex from $V_1$. 
We refer to the faces in $F_{234}$ as \emph{good faces}.
% We call \emph{good faces} the faces in $F(G)\setminus\bar{F_1}$. The set of good faces corresponds to the union of two disjoint sets of faces. We rename these sets to relate them to their definitions. Let $ F_1^{d\geq4} =F_1 \setminus \bar{F_1}$, the faces with a vertex in $V_1$ whose degree is at least $4$ and $F_{234} = F(G) \setminus F_1$, the faces without a vertex from $V_1$.
%Every essential vertex is in at least one good face by definition. 
We classify every essential vertex as \emph{type I} if it is in exactly one good face and \emph{type II} if it is in at least two good faces. 
%necesitaremos type III???
%\alma{maybe add them before?}
For any vertex $v$ in $G$, let $N_i(v)$ be the set of neighbors of $v$ in $V_i$ for all $i \in \{1, 2, 3, 4\}$.

\begin{lemma}\label{lem:color1neighbors}
    Every essential vertex $x$ of type I has an odd degree $d(x) \geq 3$ and $|N_1(x)| = \frac{d(x)-1}{2}$. 
\end{lemma}

\begin{proof}
    Let $x \in S_4$ be an essential vertex of type I with degree $d(x)$. Since $x$ is in exactly one good face, there is exactly one vertex in $N_2(x)$ and one vertex in $N_3(x)$ that are adjacent to each other. The other $d(x) - 2$ vertices are alternating between elements in $N_1(x)$ and elements in $N_2(x) \cup N_3(x)$ starting and ending with color 1 adjacent to the good face, otherwise we contradict that vertex $x$ is of type I. This implies that $\frac{d(x) -1}{2}$ of the neighbors have color 1. 
\end{proof}

% \alma{remove this and refer to lemma 6 when needed}
% \begin{proposition}\label{prop:nocolor1neighbors}
% Every essential vertex of type I is adjacent to at least one vertex in $V_1$. 
% \end{proposition}
% \begin{proof}

% Let $x \in S_4$. We know that $d(x) \geq 3$, then it is in at least $3$ faces. If $N_1(x) = \emptyset $, then $x$ is in at least $3$ good faces. 
%    % Let $x \in S_4$. If $N_1(x) = \emptyset $, then $x$ is in at least $3$ good faces.
%    \end{proof}

We say that a connected subgraph of $G$ is  \emph{$i$-$j$-alternating} if it is properly $2$-colored with colors $i$ and $j$.

\begin{lemma}\label{lem:alternating cycle}
Let $n>4$. If $x \in S_4$ is an essential vertex and every vertex in $N_1(x)$ has degree $3$, then the vertices in $N_2(x) \cup N_3(x)$ form a $2$-$3$-alternating separating cycle $C$ in $G$. Moreover, if $x$ is type I, then $x$ has odd degree $d(x) \geq 7$.
   
\end{lemma}

\begin{proof}
Let $y \in N_1(x)$. Since $d(y)= 3$, $y$ is inside a separating triangle formed by $x$ and the endpoints of a $2$-$3$-edge.
Then there is a $2$-$3$-edge for each vertex in $N_1(x)$ and a $2$-$3$ edge for each good face incident to $x$, all those edges together form a $2$-$3$-alternating separating cycle $C$ whose inside vertices are $x \in V_4$ and all the vertices in $N_1(x)$. 
Let $\ell(C)$ be the length of $C$. Observe that $\ell(C)$ is even and at least 4, so in the case where $x$ is of type I, exactly one of the edges in $C$ is part of a good face containing $x$ and $\ell(C) - 1$ edges are part of a separating triangle that contains a vertex of color 1. Then $|N_1(x)| = \ell(C) - 1$. Note that $d(x) = |N_1(x)| + \ell(C) = 2\ell(C) - 1 \geq 2(4) - 1 = 7$.\end{proof}

\begin{observation}\label{obs:neighborscolor1}
For every vertex $x$ satisfying Lemma \ref{lem:alternating cycle}, $|N_1(x)| < |N_2(x) \cup N_3(x)|$ as shown in the proof of the lemma. Since we assume that $n_2 + n_3 < \frac{n}{3}$, we have that $|N_1(x)| < \frac{n}{3}$.
\end{observation}
\begin{observation}\label{obs:swapingcolors}
    Let $C$ be a $2$-$3$-alternating separating cycle. If we swap the colors $1$ and $4$ of the vertices inside $C$, we obtain another valid $4$-coloring of $V(G)$. 
\end{observation}

% \
% \alma{Kempe chains are maximal alternating subgraphs, which is not the case here.}
% \alma{add a short 1 line proof of the observation. refer to Kempe chains. Kempe has some papers on coloring planar graphs where he does these kinds of swaps}

%Inside $C$ there is only one vertex from $V_4$, namely $x$. $C$ is formed by an even number of $2$-$3$ edges, each of which forms a triangle with $x$. Exactly one of those triangles is a good face, the rest have to be separating triangles with a degree $3$ vertex from $V_1$ inside. Since the length of the cycle is at least $4$, there are at least $3$ degree $3$ vertices from $V_1$ inside the cycle $C$.

\begin{lemma}\label{lem:typeI}
Every essential vertex $x$ of type I is adjacent to at least one vertex in $V_1^{d\geq4}$.
\end{lemma}

\begin{proof}

% Let $x \in S_4$. We know that $d(x) \geq 3$, then it is in at least $3$ faces. If $N_1(x) = \emptyset $, then $x$ is in at least $3$ good faces. 
%    % Let $x \in S_4$. If $N_1(x) = \emptyset $, then $x$ is in at least $3$ good faces.
%    \end{proof}

  Let $x \in S_4$ be an essential vertex of type I. By Lemma \ref{lem:color1neighbors}, since $d(x) \geq 3$, $|N_1(x)| \geq 1$. Assume for a contradiction that every vertex in $N_1(x)$ has degree $3$. By Lemmas \ref{lem:color1neighbors} and \ref{lem:alternating cycle}, $x$ has odd degree $d(x) \geq 7$, $|N_1(x)| = \frac{d(x) -1}{2} \geq 3$ and there is a $2$-$3$-separating alternating cycle $C$ containing exactly one vertex of color $4$, namely $x$, and $|N_1(x)|$ vertices of color $1$ in its interior. By swapping colors 1 and 4 in $C$ (Observation \ref{obs:swapingcolors}) we obtain a valid proper $4$-coloring $\chi'$ whose color classes that has $n_1 - |N_1(x)| +1$ vertices of color 1 and $n_4 - 1 + |N_1(x)|$ vertices of color 4. This would decrease the difference between $n_1$ and $n_4$ and hence contradicts our choice of $\chi_{min}$.

We have to be careful here as color 4 may now be the largest color class in $\chi'$.
  Hence the difference would be \[(n_4 - 1 + |N_1(x)|)-(n_1 - |N_1(x)| +1)= (n_4-n_1)+2|N_1(x)|-2<\frac{n}{3},\] 

  where the inequality is valid by Observation~\ref{n1n4} and Observation \ref{obs:swapingcolors}. This still contradicts our choice of $\chi_{min}$ in view of Observation \ref{obs:swapingcolors}.
  \end{proof}
 %If the length of the cycle is at least $\frac{n}{3}$, then the swapping gives a 
% \alma{remind the initial assumptions}
%   \alma{add figure}.

% In the previous subsection, we partitioned the color classes from a proper $4$-coloring into two classes, $R = V_1 \cup V_4$ and $ B = V_2 \cup V_3 $. If $R$ is the smallest color class under this assignment, we guarantee a discrepancy of at most $\frac{n}{3}$ by Lemma \ref{lemma:n_1}. Then, for this section, we assume otherwise. 
%We want to reduce the size of the largest color class $R$ to reduce the discrepancy. We use that to bound the discrepancy.
%\ref{thmDiscrepancy},  

% We start by assuming some specific properties in the $4$-coloring.
% \begin{itemize}
% \item The largest color class has size $n_1 \geq \frac{n}{2}$. Otherwise by Thm 
% \item The discrepancy of the coloring is minimum over all $4$-colorings of $G$.
% \end{itemize}

Similar to the partition of $V_4$ in essential and non-essential vertices, we partition $V_1$ according to the faces they belong to.
Let $S_1 \subseteq V_1$ be the set of vertices in $V_1$ belonging to at least one face that does not contain a vertex from $S_4$. Let $s_1 = |S_1|$. Let $\bar{S}_1 = V_1 \setminus S_1$ and call these vertices \emph{non-essential} since they can be colored red or blue without violating the polychromatic 2-coloring property. 

\begin{observation}\label{obs:deg4neighbors}
    If a vertex $y \in V_1$ has degree $4$ and two neighbors in $S_4$, then the four faces that contain $y$ contain a vertex in $S_4$ and thus $y$ is non-essential. 
\end{observation}

\begin{figure}
    \begin{center}
    \includegraphics{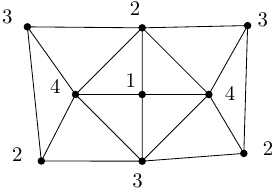}
    \caption{Illustration for Observation \ref{obs:deg4neighbors}}
    \end{center}
\end{figure}
In the following Lemma we state an upper bound for $s_1 + s_4$, which we later use to bound the discrepancy of a polychromatic $2$-coloring.

\begin{lemma}\label{lem:s1s4}
$s_1 + s_4 \leq \frac{5n-8}{7}$.
\end{lemma}

\begin{proof}

We use a discharging argument. We assign $1$ unit of charge to every face in $F_{234}$. For every vertex $y \in V_1$, we assign a charge equal to $d(y)$. Thus, the total initial charge is $|F_{234}| + |F_1| = 2n - 4$.
Apply the transfer of charge as follows: Each face in $F_{234}$ transfers its charge to its unique vertex in $V_4$.
For every vertex $y \in V_1$, if $d(y) \geq 4$, then transfer $1$ unit of charge to each of its neighbors that is in $S_4$.
The total final charge is equal to the total initial charge. We analyze the final charge in the vertices of $S_1$ and $S_4$.

 \noindent \textbf{Claim 1}. {\em Every vertex in $S_1$ has final charge at least $3$.}
 
 Let $y \in  S_1$. We know $d(y) \geq 3$. If $d(y) = 3$, then $y$ has final charge equal to $3$. If $d(y) = 4$, then $y$ has at most one neighbour in $S_4$, otherwise it contradicts Observation \ref{obs:deg4neighbors}. Therefore, the final charge is at least $3$.  If $d(y) \geq 5$, then $y$ has at most $\lfloor\frac{d(y)}{2}\rfloor$ neighbors in $S_4$, since no two color 4 vertices are adjacent. Thus, the final charge of $y$ is at least $d(y) - \lfloor\frac{d(y)}{2}\rfloor = \lceil\frac{d(y)}{2}\rceil \geq \lceil\frac{5}{2}\rceil = 3$.

\noindent \textbf{Claim 2}. {\em Every vertex in $S_4$ has final charge at least $2$.}

%Every vertex $x \in S_4$ is either of type I or of type II.
%by definition is incident to a good face, then it gets $1$ unit of charge from it. Moreover, 
Let $x \in S_4$ be an essential vertex.
If $x$ is of type II, then it gets at least $2$ units of charge from the two faces in $F_{234}$. Assume $x$ is type I, then it gets 1 unit of charge from its incident $F_{234}$ face. By Lemma \ref{lem:typeI}, $x$ is adjacent to a vertex $y \in V_1$ such that $d(y)\geq 4$, from which it gets an additional unit of charge. Thus, every vertex in $S_4$ has final charge at least $2$.

Since the total final charge is equal to the total initial charge, the final charge in $S_1 \cup S_4$ is at most the total initial charge.
\begin{equation}\label{eqn:discharging2}
3s_1 + 2s_4 \leq 2n-4 
\end{equation}

We know that $S_1 \subseteq V_1$ and $S_4 \subseteq V_4$,  then  $s_1 \leq n_1$ and $s_4 \leq n_4$. By Observation \ref{obs:smallerclass}, we have $n_4  \leq \frac{n - n_1}{3}$, which implies that:
\begin{equation}\label{eqn:obs}
s_1 + 3s_4 \leq n_1 +3\left(\frac{n - n_1}{3}\right) = n 
\end{equation}

Multiplying by 2 both sides of \eqref{eqn:discharging2} and adding it to \eqref{eqn:obs}, we get $7s_1 +7s_4 \leq 5n - 8$. Then $s_1 + s_4 \leq \frac{5n - 8}{7}$.
\end{proof}

We know that $S_1 \cup S_4 \subseteq V_1 \cup V_4$, and $|V_1 \cup V_4| = n_1 + n_4 > \frac{n}{2}$. If $s_1 + s_4 \geq \frac{n}{2}$, then let $\rho$ be the $2$-coloring that assigns $R = S_1 \cup S_4$. 
%and $B = V_2 \cup V_3 \cup \bar{S_1} \cup \bar{S_4}$. If $s_1 + s_4 < \frac{n}{2}$, then 
Otherwise, add some vertices from $\bar{S_1} \cup \bar{S_4}$ to $R$ in order to guarantee $r = |R| \geq \frac{n}{2}$. The sizes of the color classes of $\rho$ are $r = \max(s_1 + s_4, \frac{n}{2})$ and $b = n - r$. Observe that $\rho$ is a polychromatic $2$-coloring and $disc(\rho) = 2r-n$. 
By Lemma \ref{lem:s1s4}, $\operatorname{disc}(\rho) = 2r - n \leq 2\left(\frac{5n-8}{7}\right)-n = \frac{3n - 16}{7}$. This finishes the proof of Theorem~\ref{thm:3/7}.

\vspace{20pt}
\noindent{\bf Remark:}  The only place that we use the minimality of our 4-coloring is the contradictory argument in the proof of Lemma~\ref{lem:typeI} where a color-4 vertex $x$ and its color-1 neighbors are inside a 2-3-separating cycle formed by its color-2 and color-3 neighbors. Any such cycle is identified by the neighbors of a color-4 vertex. To get an algorithm, we could start from an arbitrary 4-coloring (that can be computed in
%$O(n^2)$ time~\cite{robertson1997four}
$O(n\log n)$ time \cite{Inoue2026}) and then for each such color-4 vertex $x$ we swap the color of $x$ with its color-1 neighhbors. This reduces $n_1-n_4$. This extra step takes a total of $O(n)$ time for all such $x$'s.

\section{Linear-time Algorithm to Compute a Polychromatic $2$-Coloring with Bounded Discrepancy for Planar Triangulations}\label{sectionAlgorithm}

In this section, we show how to compute, for a given planar triangulation, a polychromatic $2$-coloring whose discrepancy is at most $\frac{5n-24}{7}$, in linear time. We use the linear-time algorithm by Bose et al.~\cite{bose2003worst} to compute an initial polychromatic $2$-coloring. Their algorithm uses a linear time algorithm to compute a maximum matching in the dual graph of a triangulation by Biedl et al.~\cite{biedl2001efficient}. Thus, for any given planar graph, a polychromatic $2$-coloring can be computed in linear time. However, this polychromatic $2$-coloring may not have a guarantee on the discrepancy. We show how to extend this algorithm to guarantee a discrepancy of at most $\frac{5n-24}{7}$, without increasing the time complexity. The main approach is to make local modifications. We start by showing that the discrepancy of any polychromatic $2$-coloring on the vertices of a triangulation $G$ has an upper bound in terms of the maximum degree of $G$.

\begin{theorem}
    Let $G$ be a planar triangulation with $n \geq 3$ vertices and maximum degree $\Delta(G)$. Every polychromatic $2$-coloring of $G$ has discrepancy at most  $\frac{(\Delta(G) - 4)n + 8}{\Delta(G)}$.
\end{theorem}
\begin{proof} 

Let $G$ be a planar triangulation with $n$ vertices, and $\chi: V(G) \rightarrow \{c_1, c_2\}$ a polychromatic $2$-coloring of $G$. By Proposition \ref{propEuler} the number of faces of $G$ is exactly $2n-4$.  Since $\chi$ is a polychromatic coloring, every face contains at least one vertex of each color class. The number of faces that contain a vertex $v$ is the degree of the vertex, $d(v)$. Then, any vertex belongs to at most $\Delta(G)$ faces. Therefore, every color class of $\chi$ has size at least $\frac{2n-4}{\Delta(G)}$. This implies that $\operatorname{disc}(\chi) \leq  n - 2\left(\frac{2n-4}{\Delta(G)}\right) = \frac{(\Delta(G) - 4)n + 8}{\Delta(G)}.$\end{proof}

% \begin{corollary}

% \end{corollary}

To bound the discrepancy after a recoloring process, we define a class of colorings in which there is no single vertex to recolor that decreases the discrepancy and preserves polychromaticity. A polychromatic $2$-coloring $\rho$ of the vertices of a triangulation $G$ is \emph{locally minimal} if there is no vertex in $V(G)$ that can be recolored to reduce the discrepancy of $\rho$ without producing a monochromatic face. This naturally raises the question: How large can the discrepancy of a locally minimal coloring be?

\begin{figure}[ht]
    \centering
    \includegraphics[scale = 0.6]{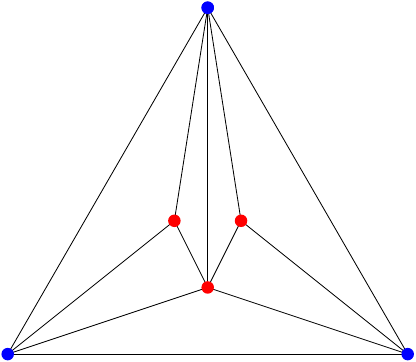}
    \caption{Construction that can be repeated in every face of a triangulation to get a locally minimal coloring. Blue points represent vertices in $B$ and red points represent vertices in $R$.}
    \label{fig:balanceA}
\end{figure}

\begin{theorem}\label{minimal}

Let $G$ be a planar triangulation with $n \geq 6$ vertices. Every polychromatic $2$-coloring of $G$ that is locally minimal has discrepancy at most $\frac{5n-24}{7}$.
    
\end{theorem}
\begin{proof}
We start with the upper bound on the discrepancy. Consider an arbitrary planar triangulation $G$ and a locally minimal polychromatic $2$-coloring $\hat{\rho}$ as described before. Let $R$, $B$ be the respective color classes, and let $r$, $b$ be their respective sizes. Without loss of generality, assume $r \geq b$. If $r\le b+1$ then the discrepancy is at most 1. Assume that $r\ge b+2$. Due to minimality of $\hat{\rho}$, every vertex $v \in R$ is incident to a face with an edge $e= (x,y)$ such that both $x$ and $y$ are in $R$, because otherwise we could recolor $v$ to reduce discrepancy. Each such edge $e$ is in exactly two faces of $G$, and hence count for at most two vertices in $R$. That means that the number of monochromatic edges of color $B$ is at least $\frac{r}{2}$.
Let $G_2$ be the subgraph of $G$ induced by $B$. $G_2$ is a planar graph with $b$ vertices, then by Proposition \ref{propEuler} it has at most $3b - 6$ edges. Therefore, there are no more than $3b - 6$ monochromatic edges of color $B$. It follows that $\frac{r}{2}  \leq  3b - 6.$ Since $r + b = n$, we have that $7b - 12 \geq n$, or $b \geq \frac{n+12}{7}$.
Thus $r = n-b \leq n-  \frac{n+12}{7}= \frac{6n-12}{7}$.
The discrepancy $r - b$ is at most $\frac{6n-12}{7} - \frac{n+12}{7} = \frac{5n-24}{7}$.

% This implies that we can use the linear time algorithm given in \cite{bose2003worst} to compute a polychromatic $2$-coloring, and if it is not minimal, a vertex from the largest color class can be changed into the smallest one, decreasing the balance. We repeat this until we get a minimal coloring, and we guarantee the discrepancy to be no more than $\frac{5n-24}{7}$. For each vertex, deciding if it can be changed colors and doing it if it is the case, takes constant time. Therefore, adding this step to the algorithm does not affect its complexity.

We show that this bound is tight by constructing an infinite family of triangulations with a locally minimal polychromatic $2$-coloring whose discrepancy is $\frac{5n-24}{7}$.  We start with $G_2$ equal to a planar triangulation on $b$ vertices of color $B$. Into each triangular face of $G_2$, we place three vertices of color $R$,
connected as shown (for one face) in Figure \ref{fig:balanceA}.
Note that no $R$-vertex can be changed to a $B$-vertex in any instance of this construction,
as each is adjacent to an edge of $G_2$; thus the construction is minimal. With $b$ vertices, $G_2$ has $2b - 4$ triangular faces, by Proposition \ref{propEuler}.
This means that $r = 3(2b - 4) = 6b - 12$, and so the construction has $r > b$
for each $b \geq 3$.
Since $r + b = n$, we have $7b - 12 = n,$ or $b = \frac{n+12}{7}$.
Thus, $r = \frac{6n-12}{7}$, as in the upper bound.
The discrepancy $r-b$ is then $\frac{5n-24}{7}$.
So we meet the upper bound by starting with any triangulation of any size.
\end{proof}

\begin{theorem}\label{thm:linear_algorithm}
    Let $G$ be a planar triangulation with $n\geq 6$ vertices. A polychromatic $2$-coloring on the vertices of $G$ whose discrepancy is at most $\frac{5n-24}{7}$ can be computed in $O(n)$ time complexity.
\end{theorem}

\begin{proof}
The proof follows from the following algorithm.
\begin{algorithm}

\begin{algorithmic}
    
    \REQUIRE A planar triangulation $G$, with $n$ vertices.
    \ENSURE A locally minimal polychromatic $2$-coloring on the vertices of $G$.
    
        \STATE Compute a polychromatic $2$-coloring. Let $R$ and $B$ be the two color classes. Without loss of generality assume $|R| \geq |B|$.
        \WHILE{there is a vertex $v \in R$, such that $v$ is not in a face with two vertices from $B$}
        \STATE Recolor $v$ from red to blue and update the color classes to $R\backslash\{v\}$ and $B \cup \{v\}$.
        \ENDWHILE
        \STATE Output the two color classes.

\end{algorithmic}
\caption{Algorithm to compute a minimal polychromatic $2$-coloring }
\label{alg:my-algorithm}
\end{algorithm}

%\alma{Mention the complexity of every step by referencing properly and conclude the subsection}

 % with no special properties related to $\operatorname{disc}(\chi)$. 
%We will show how the discrepancy is also bounded.  

\subparagraph*{Correctness.} The first step is to run the algorithm provided in \cite{bose2003worst}, to compute a coloring with no monochromatic faces. The algorithm always maintains a valid polychromatic $2$-coloring, since it only recolors a vertex when it does not produce monochromatic faces. We only recolor vertices from red to blue and every vertex is recolored at most once. Then the algorithm terminates, and the discrepancy decreases in every recoloring step until the coloring is locally minimal.

 \subparagraph*{Running Time.} Each iteration of the \textbf{while} loop examines a vertex and possibly changes its color. We say a vertex $v \in V$ is \emph{recolorable} if $v \in R$, and $v$ is not in a face with two vertices from $B$.  We can process the vertices in any given order, since a vertex that is not recolorable never becomes recolorable. For every vertex $v$, we can verify if it is recolorable in $O(d(v))$ because we evaluate all the faces that contain $v$. Since every face is triangular, during the whole process we check every face three times, and since $G$ is planar, the number of faces is $3n-6$. Therefore, the total run of the algorithm takes $O(n)$ assuming constant-time adjacency and face incidence queries. 
\end{proof}

\section{Conclusions}

%We have provided an alternative proof for two theorems from \cite{asayama2022balanced}.
%Theorem 7 and Theorem 4 
% Moreover, we have generalized these results by establishing an upper bound on the discrepancy of a polychromatic $2$-coloring in terms of an upper bound on the size of the largest color class in a proper $4$-coloring. 
%This motivates the investigation of conditions on triangulations that ensure the existence of a proper $4$-coloring with a relatively small largest color class. In particular, we focus on triangulations with a bounded independence number.

We studied the upper bounds on the discrepancy of polychromatic $2$-colorings for planar triangulations. We presented techniques that improve the previous bounds, as well as a new result that lead to the tight bound of $\frac{n}{3}$. 

One major question that remains open is whether discrepancy less than $\frac{5n-24}{7}$ can be computed in linear time. It would be interesting to employ our essential/non-essential vertex classification and our  color swapping technique of  Section~\ref{improve-sec} to improve this bound.

\section{Acknowledgement}
Part of this work was done at the 12th Annual Workshop on Geometry and
Graphs, held at Bellairs Research Institute in Barbados in February 2025. We thank the organizers
and the participants.

%{\color{blue}We are grateful to X, Y, and Z who brought the recent result of Kawarabayashi, Yoneda, and Yoneda~\cite{Kawarabayashi2026} to our attention.}

\bibliographystyle{plainurl}
\bibliography{mybibliography.bib}

\end{document}